\documentclass[11pt]{article}

\usepackage[utf8]{inputenc}
\usepackage[T1]{fontenc}
\usepackage[margin=1.05in]{geometry}
\usepackage{amsmath,amssymb,amsthm,mathtools,mathrsfs}
\usepackage{bm}
\usepackage[protrusion=true,expansion=false]{microtype}
\usepackage{enumitem}
\usepackage{booktabs}
\usepackage{array}
\usepackage[round]{natbib}
\usepackage{xcolor}
\usepackage[colorlinks=true,linkcolor=blue!55!black,citecolor=blue!50!black,urlcolor=blue!55!black]{hyperref}
\newtheorem{theorem}{Theorem}[section]
\newtheorem{proposition}[theorem]{Proposition}

\newtheorem{corollary}[theorem]{Corollary}
\theoremstyle{definition}
\newtheorem{definition}[theorem]{Definition}
\newtheorem{remark}[theorem]{Remark}

\newtheorem{observation}[theorem]{Observation}

\newcommand{\R}{\mathbb{R}}

\newcommand{\E}{\mathbb{E}}
\newcommand{\Prob}{\mathbb{P}}
\newcommand{\Q}{\mathbb{Q}}
\newcommand{\Var}{\operatorname{Var}}

\newcommand{\Tr}{\operatorname{Tr}}
\newcommand{\F}{\mathcal{F}}

\newcommand{\dd}{\,\mathrm{d}}
\newcommand{\abs}[1]{\left\lvert#1\right\rvert}
\newcommand{\norm}[1]{\left\lVert#1\right\rVert}

\newcommand{\Filt}{(\F_t)_{t\in[0,T]}}
\newcommand{\G}{\mathcal{G}}

\DeclareMathOperator*{\argmax}{arg\,max}
\newcommand{\SR}{\mathrm{SR}}

\title{\bfseries Markets Are Not Random, They Are Hard to Predict\\[2pt]
\large Instrumental Probability, the $\Prob$--$\Q$ Wedge, and the Limits of Learnable Alpha}

\author{%
Miquel Noguer i Alonso\\[2pt]
\normalsize Artificial Intelligence Finance Institute}

\date{\today}

\begin{document}
\maketitle

\begin{abstract}
\noindent
Financial returns are often called ``random,'' but the word conflates ontic chance, epistemic ignorance, strategic
feedback, and model instability. This essay argues that financial markets are not random in the ontic sense in which
a quantum measurement is random. They are causal economic systems whose future is hard to predict because relevant
causes are hidden, costly to observe, strategically used, capacity constrained, and sometimes governed by a changing
law. The formal language of finance already encodes this distinction. Prices live on filtered probability spaces
because agents have partial information; derivatives are priced under a risk-neutral measure $\Q\ne\Prob$ because
pricing is an instrumental change of measure rather than a statement about the real data-generating law; and
no-arbitrage gives martingality under an equivalent pricing measure, not full predictability failure under every
real-world information set. The paper separates no-arbitrage, informational efficiency, and net exploitability;
uses the Doob decomposition to isolate risk-compensated predictable drift from martingale innovation; adds a
capacity-and-survival layer explaining why positive signals need not be scalable; relates the $\Prob$--$\Q$ wedge to
stochastic-discount-factor geometry and relative entropy; formalises filtration sufficiency, model-selection
landscape risk, and intervention-stable causality; and connects reflexivity, microstructure, and Knightian ambiguity
to a unified entropy ledger. The disciplined thesis is therefore not that markets are unknowable, nor that they are
literally random. Markets are hard to predict, and hardest exactly where prediction is costly, competitive,
self-defeating, capacity limited, or invalidated by regime change.
\end{abstract}


\newpage

\tableofcontents

% ======================================================================

\section{Introduction}
\label{sec:intro}

It is common, and sloppy, to call asset returns ``random.'' The looseness matters because it imports a metaphysics
that finance does not possess. In a companion essay we separated two notions that ordinary language conflates: a
phenomenon is \emph{ontically random} if its outcome is fixed by \emph{no} complete specification of the present,
whereas it is merely \emph{epistemically unpredictable} if its outcome is determined by the present and the laws of
motion but inaccessible to any feasible forecaster. Quantum measurement is the paradigmatic case often invoked for
the first category; more carefully, Bell's theorem rules out local hidden-variable accounts of the relevant
correlations \citep{bell1964}. This essay does not need a full interpretation of quantum mechanics. It needs only the
weaker contrast: financial prices are macroscopic economic objects generated by information, orders, balance sheets,
constraints, and institutions.

Financial markets are epistemically difficult, not ontically random. Every tick has a cause; the chain ``information
$\to$ order $\to$ execution $\to$ price'' is causal throughout, and whatever microscopic quantum indeterminism exists
in the physical substrate of traders and machines is washed out, for modelling purposes, by decoherence and
macroscopic aggregation long before it reaches a quote \citep{zurek2003}. The Born-rule sense of chance is simply the
wrong category. But the alternative is not ``easy predictability.'' If markets are epistemic, they are \emph{hard to
predict} rather than literally unpredictable -- the word the title chooses -- and \S\ref{sec:hardness} draws the line
with care, including why one channel (reflexivity) tempts the stronger word without earning it.

This essay defends and sharpens that claim along five axes.

\begin{enumerate}[label=\textup{(\arabic*)},leftmargin=2.2em]
\item \textbf{Probability is instrumental, not ontic} (\S\ref{sec:pq}). The defining practice of mathematical
finance -- pricing under an equivalent martingale measure $\Q\ne\Prob$ -- is itself the proof. A probability one
\emph{replaces} for the purpose of pricing is, by construction, not a literal description of the world.
\item \textbf{Martingales are weaker than random walks} (\S\ref{sec:martingales-not-random}). No-arbitrage implies a
conditional-mean restriction under a pricing measure; it does not imply independent increments, Gaussian returns,
constant volatility, or absence of causal structure.
\item \textbf{Market hardness has several sources} (\S\S\ref{sec:filtration-causality}--\ref{sec:knight}).
Beyond high-dimensional partial observation, markets add costly learnability, reflexivity, microstructure frictions,
capacity limits, and Knightian law instability.
\item \textbf{The claim must not overshoot} (\S\ref{sec:doob}). Markets are not \emph{perfectly} unpredictable. The
Doob decomposition isolates a predictable, risk-compensated component (the drift, the premia) from an unpredictable
martingale innovation. Only the \emph{arbitrageable} part is unpredictable by construction.
\item \textbf{Unpredictability has a measure: entropy} (\S\S\ref{sec:entropy}--\ref{sec:relent}). Conditional entropy
quantifies residual uncertainty and mutual information quantifies informational content; relative entropy turns the
$\Prob$--$\Q$ wedge of~(1) into a number, $\tfrac12\theta^2T$; and -- via the theorems of Pesin and Brudno -- the
chaotic, algorithmic, and thermodynamic notions of unpredictability of the companion essay are revealed to be this
same quantity.
\end{enumerate}

\paragraph{Contribution.} The individual results invoked below are classical -- the fundamental theorem of asset
pricing, Girsanov's change of measure, the Doob decomposition, the Grossman--Stiglitz and adaptive-markets accounts
of efficiency, the Kyle and Glosten--Milgrom microstructure models, and the Shannon/Kullback--Leibler apparatus. The
contribution here is not a new theorem but a \emph{unification}: the same distinction -- conditional-mean
unpredictability versus structured higher-moment and regime dependence -- runs through all of them, and the pieces
collect into a single \emph{entropy ledger} (\S\ref{sec:synthesis}) in which efficiency, the price of risk,
ambiguity, and modelling under ignorance appear as one family of information-theoretic quantities. The aim is
conceptual clarity and a disciplined vocabulary, not new empirical estimates.

Section~\ref{sec:hardness} separates genuine unpredictability from mere difficulty and fixes the sense in which the
title's word is meant. Section~\ref{sec:setup} fixes the probabilistic apparatus and the no-arbitrage notion.
Section~\ref{sec:ftap} states the fundamental theorem of asset pricing, and Section~\ref{sec:martingales-not-random}
prevents the common mistake of equating martingales with independent Gaussian random walks. Section~\ref{sec:noarb-efficiency}
separates no-arbitrage from full informational efficiency. Section~\ref{sec:tradability} formalises costs, capacity,
and survival. Sections~\ref{sec:pq}--\ref{sec:statistical-hardness} separate the real-world and pricing measures and
show why physical drift is statistically faint. Section~\ref{sec:doob} decomposes returns into drift and innovation.
Section~\ref{sec:filtration-causality} makes prediction explicitly relative to filtrations, model sufficiency, and
causal intervention. Sections~\ref{sec:reflexive} and \ref{sec:microstructure} explain the market-specific feedback
and order-book mechanisms that destroy scalable alpha.
Section~\ref{sec:knight} treats law instability and ambiguity. Sections~\ref{sec:entropy}--\ref{sec:relent} recast
the argument in information-theoretic terms, while Section~\ref{sec:empirical-protocol} gives a falsifiable empirical
protocol. Section~\ref{sec:synthesis} synthesises the taxonomy and \S\ref{sec:conclusion} concludes.

% ======================================================================

\section{Unpredictable versus hard to predict}
\label{sec:hardness}

A word of precision is owed by the title. \emph{Unpredictable}, used strictly, is an ontic property: an event is
unpredictable when its outcome is fixed by \emph{no} complete specification of the present, so that no agent -- with
unlimited resources, the full body of law, and access to every variable -- can do better than chance. In the narrow
sense used here, markets do not qualify. The strongest candidate for literal unpredictability is the quantum case on
an orthodox reading; Bell's theorem rules out local hidden-variable completions of the observed correlations, but the
present essay does not rest on a particular interpretation of quantum mechanics. Its finance claim is weaker and more
robust: prices are generated by macroscopic causal mechanisms, and the relevant limits are informational,
strategic, statistical, and institutional.

The clean diagnostic is conditional information (formally, \S\ref{sec:entropy}). If a future return $r_{t+h}$ carries
no information from a filtration $\F_t$, then $I(r_{t+h};\F_t)=0$. If information exists but is small, costly,
unstable, or destroyed by use, the problem is not ontic randomness but hardness. Markets visibly have such
information: volatility is forecastable, premia exist, order flow is structured, liquidity is state-dependent, and
private information matters. What competition tends to remove is much narrower: scalable, risk-adjusted,
net-of-cost conditional mean alpha.

The hardness comes in four textures, which the body of the paper makes precise.

\begin{enumerate}[label=\textup{(\alph*)},leftmargin=2.2em]
\item \emph{Statistically hard.} The risk premium is forecastable but buried in noise; distinguishing $\Prob$ from
$\Q$ -- estimating expected returns -- requires horizons of order $\theta^{-2}$, often decades of data
(\S\ref{sec:relent}; \citealp{merton1980}). Low signal-to-noise is already enough to defeat many apparently obvious
signals.
\item \emph{Learnability hard.} Even when an oracle conditional mean exists, a finite-sample learner must estimate it
from dependent, non-stationary, multiply searched data. The signal may decay before its estimation error falls below
its economic value (\S\ref{sec:learnability}).
\item \emph{Endogenously hard.} A market's difficulty is not merely exogenous, like a Lyapunov horizon in chaos. Any
edge you find and trade changes prices, impact, capacity, and the incentives of others (\S\ref{sec:reflexive}). The
unpredictability is manufactured and held near the break-even margin by the predictors themselves.
\item \emph{Not probabilizable.} Under regime change the target law itself is not fixed (\S\ref{sec:knight}); here
``hard to predict'' already presupposes a stable distribution to pin down -- exactly what Knightian uncertainty
denies.
\end{enumerate}

That ``unpredictable'' cannot be literal is guaranteed by the informational equilibrium itself. Were the arbitraged
residual truly unpredictable \emph{to everyone}, no one would pay to gather information, and then -- by
\citet{grossmanstiglitz1980} -- prices could not be efficient. Efficiency requires that someone can predict something
well enough to pay for price discovery. Unpredictability in a market is therefore always relative to an information
set and net of costs: unpredictable to the marginal public trader, predictable to the owner of a private signal, a
faster pipe, a better model, or a risk tolerance others do not possess.

We therefore choose the precise word. Markets are \emph{hard to predict}, measurably so, in the textures above. One
texture comes close to deserving the stronger word -- self-defeating prediction at the competitive margin, which
mimics a wall without being one -- and ``close'' is exactly why ``unpredictable'' is a tempting overstatement. The
honest long form of the title is this: \emph{markets are not random; they are hard to predict, and at the margin
where prediction matters most, hard in a way that defeats the act of prediction.}

% ======================================================================
\section{The apparatus: information, strategies, arbitrage}
\label{sec:setup}

Fix a finite horizon $T$ and a filtered probability space $(\Omega,\F,\Filt,\Prob)$ satisfying the usual conditions.
The filtration $\Filt$ \emph{is} the information structure: $\F_t$ is what is known at time $t$, and a process is
\emph{adapted} if it is observable as it unfolds, \emph{predictable} if its time-$t$ value is known strictly before
$t$ (i.e.\ $\F_{t^-}$-measurable). A predictable process is the formal object of a trading rule: a position must be
chosen on the basis of past information, not contemporaneous surprise.

Let $S=(S_t)_{t\in[0,T]}$ be an $\R^d$-valued semimartingale of discounted asset prices (numéraire $B_t$, the
money-market account). A \emph{self-financing strategy} is a predictable, $S$-integrable process $H$ with discounted
wealth
\begin{equation}\label{eq:wealth}
V_t(H)\;=\;V_0+\int_0^t H_u\dd S_u ,
\end{equation}
the stochastic integral being the cumulative gains from trade. The two ingredients of the apparatus -- that
strategies are \emph{predictable} and that gains are an \emph{integral against the price} -- already encode the
asymmetry between what can and cannot be forecast: $H$ uses the past, $\dd S$ delivers the surprise.

\begin{definition}[Arbitrage and NFLVR]\label{def:arb}
An \emph{arbitrage} is a self-financing $H$ with $V_0=0$, $V_T(H)\ge0$ a.s., and $\Prob(V_T(H)>0)>0$: a risk-free
profit from nothing. The market satisfies \emph{No Free Lunch with Vanishing Risk} \citep{delbaenschachermayer1994}
if no sequence of admissible terminal gains converges uniformly to a nonnegative, nontrivial payoff -- the correct
closure of ``no arbitrage'' in continuous time.
\end{definition}

The first source of unpredictability appears already here and is the one finance \emph{shares} with statistical
physics. A market aggregates an enormous number of coupled, partially observed degrees of freedom -- positions,
beliefs, orders, inventories -- none of which any single agent can measure. As in Boltzmann's treatment of a gas,
the rational response is not to integrate the (inaccessible) microscopic equations of motion but to model the
aggregate probabilistically. This is genuine \emph{epistemic} ignorance, identical in kind to the classical
randomness of a many-body system, and it is the least market-specific of the three sources.

% ======================================================================
\section{No exploitable structure is a martingale}
\label{sec:ftap}

The economic content of ``you cannot systematically beat the market without bearing risk'' is, mathematically, a
\emph{martingale} statement. This is the bridge from informal efficiency to a precise object.

\begin{theorem}[Fundamental Theorem of Asset Pricing; \citealp{harrisonkreps1979,harrisonpliska1981,dalang1990,delbaenschachermayer1994}]
\label{thm:ftap}
Let $S$ be a locally bounded semimartingale. Then NFLVR holds if and only if there exists a probability measure
$\Q\sim\Prob$ (an \emph{equivalent local martingale measure}) under which $S$ is a local martingale. For general
semimartingales, ``local martingale'' is replaced by ``$\sigma$-martingale.''
\end{theorem}

\begin{remark}
In discrete time the statement is exact and elementary \citep{dalang1990}: absence of arbitrage is equivalent to
the existence of an equivalent martingale measure, with no local/$\sigma$ qualification.
\end{remark}

Theorem~\ref{thm:ftap} is the rigorous form of the reflexive intuition to be developed in \S\ref{sec:reflexive}: if
discounted prices admit \emph{any} predictable, risk-uncorrelated drift that a strategy could harvest, that drift is
an arbitrage, and by the theorem no equivalent martingale measure exists -- the market is not viable. Conversely, in
a viable market the discounted price is, under $\Q$, a martingale: its increments are unforecastable in the precise
sense
\begin{equation}\label{eq:mart}
\E^{\Q}\!\left[\,S_t-S_s \mid \F_s\,\right]=0,\qquad s\le t .
\end{equation}
This is exactly \citet{samuelson1965}: \emph{properly anticipated prices fluctuate randomly}. The word ``randomly''
in Samuelson's title is, read correctly, the martingale property \eqref{eq:mart} -- a statement about
\emph{conditional unforecastability given the information set}, not about ontic chance. The efficient-market
hypothesis \citep{fama1970} is the empirical claim that real prices are close to \eqref{eq:mart} under the
risk-adjusted measure.

% ======================================================================
\section{Martingales are not random walks}
\label{sec:martingales-not-random}

The martingale conclusion of Theorem~\ref{thm:ftap} is often misread. A martingale is a conditional-mean object, not a
claim of independence, normality, homoskedasticity, or metaphysical randomness. It says that the next increment has
zero conditional expectation under the relevant measure and information set. It does not say that the full
conditional distribution is constant through time, that the tails are thin, that order flow has no structure, or that
future volatility cannot be forecast.

The distinction is elementary but decisive. Let
\begin{equation}\label{eq:mds-vol}
r_{t+1}=\sigma_t\epsilon_{t+1},\qquad
\E[\epsilon_{t+1}\mid\F_t]=0,
\qquad
\E[\epsilon_{t+1}^2\mid\F_t]=1,
\end{equation}
where $\sigma_t$ is positive and $\F_t$-measurable. Then $(r_{t+1})$ is a martingale-difference sequence because
$\E[r_{t+1}\mid\F_t]=0$. Yet
\begin{equation}\label{eq:vol-forecastable}
\Var(r_{t+1}\mid\F_t)=\sigma_t^2,
\end{equation}
so the scale of the next move is forecastable if $\sigma_t$ is known. ARCH and GARCH models are exactly of this kind:
the direction may be close to conditionally mean-zero while volatility is highly predictable
\citep{engle1982,bollerslev1986}. Heavy tails, volatility clustering, leverage effects, jumps, and microstructure
frictions all live comfortably inside a martingale framework.

Thus a random walk with independent increments is only a special martingale; variance-ratio tests reject the former
for many return series while remaining consistent with the latter \citep{lomackinlay1988,leroy1989}. A martingale with stochastic volatility,
time-varying jump intensity, endogenous liquidity, and conditional skewness is not a random walk in any plain-language
sense. It is a causally structured process whose first conditional moment has been neutralised by pricing,
competition, or change of measure:
\begin{equation}\label{eq:mean-not-law}
\E^{\Q}[\Delta S_{t+1}\mid\F_t]=0
\quad\not\Rightarrow\quad
\mathcal L^{\Q}(\Delta S_{t+1}\mid\F_t)
\text{ is constant, independent, or Gaussian.}
\end{equation}
The correct inference from no-arbitrage is not ``prices are random.'' It is that the easiest arbitrage channel -- the
predictable, scalable, risk-adjusted mean -- is closed under the relevant pricing representation.

% ======================================================================
\section{No-arbitrage is not the same as full efficiency}
\label{sec:noarb-efficiency}

A useful extension is also a warning. Theorem~\ref{thm:ftap} is sometimes read as if it said that markets are
unpredictable in the real world. It says something narrower and more powerful: in a viable market there exists at
least one equivalent pricing measure under which discounted prices are martingales. It does \emph{not} say that the
physical conditional mean under $\Prob$ is zero, that every investor has the same information, that transaction costs
are absent, or that every statistically detectable pattern is exploitable. Confusing these layers is the fastest way
to turn a correct no-arbitrage theorem into an exaggerated philosophy of randomness.

There are four distinct statements.
\begin{enumerate}[label=\textup{(\roman*)},leftmargin=2.2em]
\item \emph{No-arbitrage}: no self-financing predictable strategy creates a nonnegative terminal payoff from zero
capital with positive probability of gain. This is a viability condition.
\item \emph{Risk-neutral martingality}: discounted prices are local martingales under some $\Q\sim\Prob$. This is a
pricing representation.
\item \emph{Informational efficiency}: conditional expected abnormal returns are zero relative to a specified
information set and benchmark model. This is an empirical claim \citep{fama1970}.
\item \emph{Economic exploitability}: a signal survives financing costs, bid--ask spreads, market impact, borrow
fees, latency, taxes, capital constraints, model error, and capacity. This is an implementation claim.
\end{enumerate}
Only the first two are mathematical consequences of the idealised model. The last two are facts about a particular
market, date, information set, and trading technology.

A compact way to express the distinction is to write a net conditional alpha for a predictable strategy $H$ over a
small horizon $\Delta t$:
\begin{equation}\label{eq:netalpha}
\alpha^{\mathrm{net}}_t(H)
=
\E^{\Prob}\!\left[\Delta V_{t+\Delta t}(H)-r_t V_t(H)\Delta t\mid\F_t\right]
-
C_t(H)-I_t(H)-K_t(H),
\end{equation}
where $C_t(H)$ denotes explicit trading costs, $I_t(H)$ market impact, and $K_t(H)$ the shadow cost of capital,
funding, leverage, and risk limits. The martingale ideal is not that every statistical pattern has zero raw
conditional mean; it is that any positive, scalable, risk-adjusted, \emph{net} alpha is competed down until
\begin{equation}\label{eq:netalphazero}
\alpha^{\mathrm{net}}_t(H)\le 0
\end{equation}
for the marginal trader who is able to deploy it. Limits to arbitrage make this inequality local rather than
absolute: a signal may persist because it is too small, too crowded, too slow, too capital intensive, too risky, or
too costly to exploit at scale \citep{shleifervishny1997}.

This qualification strengthens the essay rather than weakening it. Financial predictability is not absent; it is
\emph{priced, rationed, and capacity constrained}. A premium can be predictable because it compensates risk. A
microstructure pattern can be predictable because only the fastest or best-capitalised traders can harvest it. A
long-horizon valuation signal can be predictable but untradable at attractive leverage because the drawdown horizon
is longer than the investor's survival horizon. None of these cases is ontic randomness. They are precisely the
terrain of epistemic, strategic, and institutional hardness.

\begin{remark}[The information-set relativity of efficiency]
The phrase ``the market is efficient'' is incomplete unless it specifies the information set. A process can be a
martingale relative to public daily information and fail to be a martingale relative to private order flow,
exchange-colocated messages, dealer inventory, or the latent state inferred by a superior model. Efficiency is always
a statement about $(S,\F,\Prob,\text{costs})$, not about $S$ alone.
\end{remark}

% ======================================================================
\section{From predictability to tradability: costs, capacity, and survival}
\label{sec:tradability}

The previous section separates raw predictability from economic exploitability. The distinction deserves its own
formal layer because many empirical disputes about ``market randomness'' are really disputes about this layer. A
signal may forecast returns, but the act of turning it into a position introduces spread costs, market impact,
financing, risk limits, and crowding. These frictions are not minor implementation details; they are part of the
mechanism by which a market converts public predictability into low net alpha.

Let $m_t=\E^{\Prob}[r_{t+1}\mid\G_t]$ be the conditional excess-return forecast available under an enlarged
filtration $\G_t$, and let $q_t$ be the signed dollar position chosen before $r_{t+1}$ is realised. A one-period
certainty-equivalent objective can be written as
\begin{equation}\label{eq:capacity-objective}
\Pi_t(q)
=
q\,m_t
- C_t(q)
- \frac{\gamma_t}{2}q^2\sigma_t^2
- F_t(q),
\end{equation}
where $C_t(q)$ is explicit and implicit trading cost, $\gamma_t q^2\sigma_t^2/2$ is the local risk-capital charge,
and $F_t(q)$ denotes funding, borrow, inventory, leverage, or mandate costs. A stylised cost function consistent
with execution models is
\begin{equation}\label{eq:cost-function}
C_t(q)
=
\frac{s_t}{2}|q|+\eta_t |q|^{1+\delta}+\lambda_t q^2,
\qquad \delta\in(0,1],
\end{equation}
where $s_t$ is the bid--ask spread, the middle term captures nonlinear temporary impact, and the quadratic term
captures inventory and permanent-impact effects. Linear-quadratic versions of this logic underlie optimal execution
models such as \citet{almgrenchriss2001}.

\begin{proposition}[Finite capacity of alpha]\label{prop:capacity}
Suppose $C_t$ and $F_t$ are convex, increasing in $|q|$, and not identically zero. A positive forecast $m_t>0$ can
produce positive marginal value at small size while having a finite optimal capacity $q_t^\star$. If $q_t^\star$ is
interior, it satisfies
\begin{equation}\label{eq:capacity-foc}
m_t
=
C_t'(q_t^\star)+\gamma_t q_t^\star\sigma_t^2+F_t'(q_t^\star).
\end{equation}
Thus the signal is not refuted by finite capacity: the predictable component is real, but its scalable value is
absorbed by impact, risk capital, and institutional constraints.
\end{proposition}

\begin{proof}
The objective \eqref{eq:capacity-objective} is concave in $q$: $C_t$ and $F_t$ are convex and the risk term
$-\tfrac{\gamma_t}{2}q^2\sigma_t^2$ is concave, so a maximiser exists. At an interior optimum $\Pi_t'(q_t^\star)=0$,
which is \eqref{eq:capacity-foc}. Since marginal cost $C_t'$ and the risk charge $\gamma_t q\sigma_t^2$ increase
without bound in $|q|$ while $m_t$ is fixed, the optimal $q_t^\star$ is finite.
\end{proof}

The proposition explains why ``positive expected return'' is too weak a standard. A valuation signal may be correct
but slow; a high-frequency signal may be fast but capacity constrained; a short-rebate signal may disappear after
borrow fees; a convergence trade may be profitable in expectation but impossible to lever through a crisis. Economic
alpha is therefore a fixed point of forecast strength, cost technology, capital supply, and competitor behaviour.

Survival adds another constraint. Let $W_t$ be the wealth of the strategy and let $L_t$ be a capital, margin, or
redemption boundary. Define the survival time
\begin{equation}\label{eq:survival-time}
\tau_L=\inf\{t\le T: W_t\le L_t\}.
\end{equation}
A strategy with positive long-run expectation is institutionally viable only if
\begin{equation}\label{eq:survival-constraint}
\Prob(\tau_L\le T)\le \varepsilon
\end{equation}
for the capital provider's tolerance $\varepsilon$. This is the finite-horizon version of limits to arbitrage
\citep{shleifervishny1997}: the market can remain predictably mispriced longer than the arbitrageur can remain
solvent. Hard-to-predict markets are therefore not only statistically hard; they are also financially hard.

% ======================================================================
\section{Probability in finance is instrumental: the \texorpdfstring{$\Prob$--$\Q$}{P--Q} wedge}
\label{sec:pq}

Here is the sharpest argument that markets are not ``random'' in any ontic sense: the central computation of the
field is performed under a probability measure that practitioners \emph{know} to be false as a description of the
world, and they swap it in deliberately. A law you replace for convenience cannot be a law of nature.

Make it concrete in the Black--Scholes--Merton model \citep{blackscholes1973,merton1973}. Under the physical measure
$\Prob$ a stock follows geometric Brownian motion
\begin{equation}\label{eq:gbmP}
\dd S_t = S_t\!\left(\mu\dd t+\sigma\dd W_t^{\Prob}\right),
\end{equation}
with real expected return $\mu$. Let $r$ be the risk-free rate and define the \emph{market price of risk}
$\theta=(\mu-r)/\sigma$. By Girsanov's theorem \citep{girsanov1960}, the density
\begin{equation}\label{eq:girsanov}
\frac{\dd\Q}{\dd\Prob}\bigg|_{\F_T}=\exp\!\Bigl(-\theta\,W_T^{\Prob}-\tfrac12\theta^2 T\Bigr)
\end{equation}
defines an equivalent measure $\Q\sim\Prob$ under which $W_t^{\Q}=W_t^{\Prob}+\theta t$ is a Brownian motion, and
\eqref{eq:gbmP} becomes
\begin{equation}\label{eq:gbmQ}
\dd S_t = S_t\!\left(r\dd t+\sigma\dd W_t^{\Q}\right),
\qquad\text{so that}\qquad e^{-rt}S_t \text{ is a } \Q\text{-martingale.}
\end{equation}

\begin{proposition}[Pricing erases the premium]\label{prop:erase}
Under $\Q$ every asset earns the risk-free rate in expectation: the entire physical drift $\mu$, and with it the
equity premium $\mu-r$, is absent from \eqref{eq:gbmQ}. The premium reappears only on changing back to $\Prob$ via
\eqref{eq:girsanov}. Equivalently, in stochastic-discount-factor form \citep{cochrane2005}, every gross return $R$
satisfies
\begin{equation}\label{eq:sdf}
\E^{\Prob}\!\left[\,M\,R\,\right]=1,\qquad M=e^{-r}\,\frac{\dd\Q}{\dd\Prob},
\end{equation}
so the object that prices ($M$, equivalently $\Q$) is distinct from the object that describes realised returns
($\Prob$), and their wedge $\dd\Q/\dd\Prob$ \emph{is} the price of risk.
\end{proposition}

\paragraph{The stochastic-discount-factor geometry.}
The same wedge has a geometric form. For any excess return $R^e$, no-arbitrage implies
\begin{equation}\label{eq:sdf-excess}
\E^{\Prob}[M R^e]=0.
\end{equation}
Since $\E^{\Prob}[M]>0$, \eqref{eq:sdf-excess} gives
\begin{equation}\label{eq:hj-bound}
\frac{|\E^{\Prob}[R^e]|}{\sqrt{\Var^{\Prob}(R^e)}}
\le
\frac{\sqrt{\Var^{\Prob}(M)}}{\E^{\Prob}[M]},
\end{equation}
the Hansen--Jagannathan bound \citep{hansenjagannathan1991}. High Sharpe ratios therefore require a volatile
stochastic discount factor. In words: a large predictable premium is not a free abnormality; it must be matched by a
large covariance with marginal value, unless the model of $M$ or the information set is wrong. The $\Prob$--$\Q$
wedge is therefore simultaneously a change of measure, a price-of-risk object, and a geometry of admissible Sharpe
ratios.

The interpretive force is immediate. We compute prices under $\Q$ \emph{because} it is the measure that makes
discounted prices martingales -- a mathematical convenience guaranteed to exist by Theorem~\ref{thm:ftap} -- and
under $\Q$, by construction, expected returns carry no information about real expected returns. No one believes the
world evolves under $\Q$; $\Q$ is a pricing device. A discipline whose core quantity is a probability measure
chosen \emph{for its convenience and known to misdescribe the data} is not treating probability as a feature of
reality. It is treating it as an instrument. Ontic randomness is not even in the conceptual vocabulary; what the
machinery encodes is a structured account of risk and ignorance.

\begin{remark}[Contrast with the quantum case]
In operational quantum mechanics the Born probabilities are not swappable in the way pricing measures are: the rule
$\Prob(\cdot)=\norm{\Pi\psi}^2$ is part of the theory's empirical content, and Bell's theorem rules out a local
hidden-variable replacement. In finance there is a whole equivalence class of measures $\{\Q\sim\Prob\}$ (a
continuum, in incomplete markets) and we select among them by modelling convention, calibration target, utility, or
entropy criterion. The very multiplicity of admissible probabilities is a signature of instrumentality.
\end{remark}

% ======================================================================
\section{Statistical hardness and learnability: why drift is almost invisible}
\label{sec:statistical-hardness}
\label{sec:learnability}

The $\Prob$--$\Q$ distinction is conceptually sharp, but empirically faint. Volatility can be estimated from many
short-horizon observations; drift cannot be made precise merely by sampling more finely over the same calendar time.
This is why expected returns are the hardest object in empirical finance \citep{merton1980}.

In the simplest diffusion model
\begin{equation}\label{eq:drift-diffusion}
\dd X_t=\mu_e\dd t+\sigma\dd W_t,
\end{equation}
where $\mu_e$ is the physical excess drift, observing the path on $[0,T]$ gives an estimator with standard error of
order $\sigma/\sqrt{T}$. The signal-to-noise ratio is therefore
\begin{equation}\label{eq:drift-snr}
\frac{|\mu_e|}{\sigma/\sqrt{T}}=|\theta|\sqrt{T},
\qquad \theta=\frac{\mu_e}{\sigma}.
\end{equation}
The key point is calendar time: increasing the sampling frequency inside a fixed year improves the estimate of
quadratic variation, but it does not create more independent evidence about the annual drift. To test a drift with
two-sided size $\alpha$ and power $1-\beta$, the approximate required horizon is
\begin{equation}\label{eq:drift-horizon}
T\gtrsim
\left(\frac{z_{1-\alpha/2}+z_{1-\beta}}{|\theta|}\right)^2.
\end{equation}
For an annual Sharpe ratio $\theta=0.33$, even ordinary $5\%$ significance and $80\%$ power imply a horizon on the
order of decades. For the far smaller post-cost alphas that matter in liquid markets, the required horizon is longer
still.

This is the statistical version of the essay's thesis. The physical drift is not absent. It is simply small relative
to diffusion noise, and the information distance between $\Prob$ and $\Q$ accumulates slowly. Hence two claims can
both be true: risk premia exist and are economically central, but their estimation is so slow that short samples look
nearly martingale. The market is not ontically random; it is statistically low-signal.

% ======================================================================
\section{What \emph{is} predictable: the Doob decomposition}
\label{sec:doob}

Before pressing the unpredictability thesis further, we guard it against overshooting. Markets are not perfectly
unpredictable; the precise statement is that the \emph{unexploitable} part of returns is unpredictable, while a
\emph{risk-compensated} part is predictable and is meant to be. The clean separator is the Doob decomposition.

See \citet{doob1953} and, in continuous time, the Doob--Meyer version in \citet{protter2005}.
\begin{theorem}[Doob decomposition]
\label{thm:doob}
Let $(X_n)_{n\ge0}$ be integrable and adapted to $(\F_n)$. Then $X$ admits a unique decomposition
\begin{equation}\label{eq:doob}
X_n = X_0 + M_n + A_n,\qquad
A_n=\sum_{k=1}^{n}\E\!\left[X_k-X_{k-1}\mid\F_{k-1}\right],
\end{equation}
where $M$ is a martingale with $M_0=0$ and $A$ is predictable with $A_0=0$. If $X$ is a submartingale then $A$ is
non-decreasing.
\end{theorem}

Read \eqref{eq:doob} on a cumulative-excess-return process $X$. The predictable increment
$\Delta A_k=\E[\Delta X_k\mid\F_{k-1}]$ is the \emph{conditional expected excess return} -- the risk premium, known
one step ahead. The martingale increment $\Delta M_k=\Delta X_k-\Delta A_k$ has $\E[\Delta M_k\mid\F_{k-1}]=0$: it is
the \emph{innovation}, unforecastable by definition. Two corollaries delimit the thesis exactly.

\begin{corollary}[The trichotomy of return components]\label{cor:trichotomy}
\hfill
\begin{enumerate}[label=\textup{(\roman*)}]
\item \emph{Arbitrageable part: unpredictable by construction.} Any predictable component uncorrelated with priced
risk is, by Theorem~\ref{thm:ftap}, an arbitrage and is competed away; what remains in $\Delta M$ is conditionally
mean-zero. This is the kernel of truth in ``random walk.''
\item \emph{Risk-premium part: predictable, but not a free lunch.} $\Delta A$ is forecastable -- equity premium,
volatility dynamics, factor and value premia (\S\ref{sec:synthesis}) -- yet harvesting it requires bearing the very
risk it compensates, so it does not contradict \eqref{eq:mart} under $\Q$, where $A$ is reabsorbed into the
numéraire drift.
\item \emph{Regime/tail part: not probabilizable.} The decomposition \eqref{eq:doob} presupposes a fixed $\Prob$
under which the conditional expectations exist. When the law itself drifts (\S\ref{sec:knight}), even $A$ and $M$
are not well defined, and unpredictability passes beyond the reach of probability.
\end{enumerate}
\end{corollary}

Thus, when finance says that ``markets are unpredictable,'' the precise statement is that ``the arbitrageable component is unpredictable relative to the relevant filtration and cost structure.'' The premium
component is precisely what quantitative investing seeks; it is predictable and risky, not random and free.

% ======================================================================
\section{Filtrations, sufficiency, and causal predictability}
\label{sec:filtration-causality}

The previous sections use a filtration $\Filt$, but the choice of filtration is not innocent. It is the formal
version of what the modeller, trader, regulator, or econometrician is allowed to know. A return can be unpredictable
relative to one filtration and predictable relative to a larger one. The same price process may therefore be a
martingale for the public investor, a submartingale for the informed trader, and a nearly deterministic execution
problem for the market maker observing the order book at microsecond frequency.

\begin{definition}[Filtration enlargement and sufficiency]
Let $\F_t\subseteq\G_t$ be two information sets and let $Y_{t+h}$ be a future payoff or return. The smaller
filtration $\F_t$ is \emph{sufficient for predicting $Y_{t+h}$ relative to $\G_t$} if
\begin{equation}\label{eq:filtration-sufficiency}
\Prob(Y_{t+h}\in A\mid\G_t)=\Prob(Y_{t+h}\in A\mid\F_t)
\qquad\text{for all measurable }A.
\end{equation}
Equivalently, $Y_{t+h}$ is conditionally independent of the extra information $\G_t\setminus\F_t$ once $\F_t$ is
known. Failure of \eqref{eq:filtration-sufficiency} is not randomness; it is missing information.
\end{definition}

The mutual-information version is immediate:
\begin{equation}\label{eq:incremental-mi}
I(Y_{t+h};\G_t)-I(Y_{t+h};\F_t)
= I(Y_{t+h};\G_t\mid\F_t)\ge 0.
\end{equation}
The extra filtration improves prediction exactly by the conditional mutual information it contributes. A superior
forecasting system is, in this sense, an engineered filtration enlargement: alternative data, faster feeds, cleaner
labels, better factor construction, and deeper order-book states all attempt to make $\G_t$ strictly richer than
$\F_t$.

Prediction, however, is still not causation. Granger causality asks whether the history of $X$ improves the
conditional distribution of $Y$ after controlling for the history already in the filtration \citep{granger1969}:
\begin{equation}\label{eq:granger}
\mathcal L(Y_{t+h}\mid\F_t^Y\vee\F_t^X)
\ne
\mathcal L(Y_{t+h}\mid\F_t^Y).
\end{equation}
This is a predictive notion. A structural causal claim is stronger: it asks what would happen under an intervention,
for example if a trader actually deploys the strategy, if many traders copy it, or if a regulator changes the rule
that created the signal \citep{pearl2009}. The reflexive character of markets makes the distinction essential. A
variable may predict returns in historical data and still fail as a causal trading lever because acting on it changes
liquidity, impact, crowding, and the beliefs of other agents.

\begin{observation}[Predictive edge versus intervention-stable edge]
\label{prop:intervention-edge}
Let $m_t=\E^{\Prob}[r_{t+h}\mid\G_t]$ be a predictive signal. The signal is an intervention-stable trading edge only
if the conditional law of $r_{t+h}$ under the intervention induced by trading on $m_t$ remains close to the
historical law used to estimate $m_t$:
\begin{equation}\label{eq:edge-stability}
\mathcal L\!\bigl(r_{t+h}\mid\G_t,\operatorname{do}(H=H(m_t))\bigr)
\approx
\mathcal L\!\bigl(r_{t+h}\mid\G_t\bigr)
\end{equation}
after costs, impact, and capital constraints. When \eqref{eq:edge-stability} fails, a forecast can be statistically
valid and economically self-destroying at the same time.
\end{observation}

This is the causal form of the essay's thesis. Markets are difficult not merely because information is hidden, but
because the act of using information is itself a market event. The correct object is therefore not ``the probability
of the next return'' in isolation. It is the conditional law of the next return under a specified filtration,
strategy, cost structure, and equilibrium response.

% ----------------------------------------------------------------------
\subsection{The model-selection landscape problem}
\label{sec:model-landscape}

A richer filtration is not enough. The modeller must compress it into a statistic, representation, or trading rule.
Let $Z_t=f_\lambda(\G_t)$ be the feature representation produced by model choice $\lambda\in\Lambda$. The statistic
$Z_t$ is prediction-sufficient for $Y_{t+h}$ if
\begin{equation}\label{eq:model-sufficiency}
\mathcal L(Y_{t+h}\mid\G_t)=\mathcal L(Y_{t+h}\mid Z_t),
\qquad\text{equivalently}\qquad
I(Y_{t+h};\G_t\mid Z_t)=0.
\end{equation}
The filtration contains all available information; the model decides which part of it is operationally retained.
Model error is therefore an information bottleneck, not a cosmetic estimation issue.

The \emph{landscape problem} is the difficulty of choosing $f_\lambda$ and its hyperparameters when many plausible
representations fit the past. If $\widehat U(\lambda)=U(\lambda)+\xi_\lambda$ is an estimated utility, Sharpe ratio,
or loss improvement, then selecting
\begin{equation}\label{eq:selection-landscape}
\widehat\lambda=\argmax_{\lambda\in\Lambda}\widehat U(\lambda)
\end{equation}
selects not only genuine edge but also favourable noise. In the stylised independent Gaussian case,
\begin{equation}\label{eq:winner-bias}
\E\!\left[\max_{1\le j\le N}\xi_j\right]\approx \sigma_\xi\sqrt{2\log N},
\end{equation}
so the best in-sample strategy in a large search landscape is upward biased even when every candidate has zero true
alpha. This is the mathematical core of the factor-zoo and backtest-overfitting problem \citep{harveyliuzhu2016,lopezdeprado2018}.

The consequence is conceptual as well as empirical. A market can be predictable relative to some unknown sufficient
representation $Z_t^\star$ while remaining hard to predict for agents searching over a vast unstable landscape of
models. Limited predictability is therefore partly a representation problem: the causes may be present in $\G_t$, but
not in the compressed statistic the trader can estimate, validate, and trade before the regime changes.

% ----------------------------------------------------------------------
\subsection{Forecast combination and the entropy of ensembles}
\label{sec:forecast-combination}

A practical forecaster rarely relies on a single filtration $\F_t$; she combines signals
$S^{(1)}_t,\dots,S^{(k)}_t$ from different sources. Let
$\mathcal S_t=\sigma(S^{(1)}_t,\dots,S^{(k)}_t)$ and $\G_t=\F_t\vee\mathcal S_t$. The chain rule for mutual
information gives
\begin{align}\label{eq:ensemble-mi}
I(r_{t+h}; \G_t)
&= I(r_{t+h}; \F_t)+I(r_{t+h}; \mathcal S_t\mid \F_t)\\
&\ge I(r_{t+h}; \F_t).
\end{align}
Thus more information cannot hurt the population Bayes predictor. However, the conditional entropy satisfies
\[
H(r_{t+h}\mid \G_t) = H(r_{t+h}\mid \F_t) - I(r_{t+h}; \mathcal S_t \mid \F_t),
\]
so the marginal gain from an additional signal is its conditional mutual information given the others. In liquid
markets, the law of diminishing returns applies: after a few strong predictors, additional signals contribute
negligibly to reducing the residual entropy. This is why modern quantitative funds focus on orthogonal,
low-correlation alphas rather than many correlated ones. In finite samples, however, more signals can hurt through
estimation error and model selection; the population inequality becomes an engineering problem once $k$ is large
relative to the effective sample size.

% ======================================================================
\section{Reflexivity: the market-specific source}
\label{sec:reflexive}

The second source of unpredictability has no analogue in physics, and it is the deep one. A weather system is
indifferent to the forecast; a market is not. Prices are set by agents acting on their forecasts, so a forecast is
an \emph{input} to the very process it predicts. Predictability is therefore \emph{self-defeating}: an exploitable
regularity, once acted upon, is removed by the trading it provokes.

We can state the mechanism as a fixed-point property consistent with Theorem~\ref{thm:ftap}.

\begin{proposition}[Self-defeating predictability]\label{prop:reflexive}
Suppose prices are set by a no-arbitrage (or competitive-equilibrium) condition that is a functional of the agents'
conditional forecasts. If some publicly available forecast implied a predictable, risk-uncorrelated, positive-mean
strategy, agents would scale it without bound, and the resulting demand would move prices until the strategy's
conditional mean excess return is zero. Hence in any viable equilibrium the only predictability that survives is
risk-compensated; the arbitrageable component is, endogenously, a martingale innovation.
\end{proposition}

\begin{proof}[Sketch]
Suppose a public forecast implies a risk-uncorrelated strategy with conditional mean excess return $g_t>0$ on a set
of positive probability. Competitive agents bearing no compensating risk scale demand in $g_t$; under a price map
increasing in aggregate demand, prices adjust until $g_t=0$. Hence only risk-compensated conditional means survive,
and the arbitrageable residual is a $\Q$-martingale innovation (Theorem~\ref{thm:ftap}). The argument is an
equilibrium heuristic, valid under the same frictionless-scaling idealisation as Theorem~\ref{thm:ftap}.
\end{proof}

This is the engine behind the martingale of \S\ref{sec:ftap}: the unforecastability \eqref{eq:mart} is not a
primitive assumption about price ``noise'' but an \emph{equilibrium outcome} of agents eliminating what they can
forecast. It is also why the correct object is a no-arbitrage martingale and the billiard of classical chaos is the
wrong analogy: no Lyapunov exponent erases its own predictability, but a profit-seeking forecaster does.

Pushed to its limit the mechanism is paradoxical, and the paradox is informative.

\begin{proposition}[Impossibility of fully efficient markets; \citealp{grossmanstiglitz1980}]\label{prop:gs}
If information is costly and prices are fully revealing, then informed agents earn no rent over the uninformed and
have no incentive to gather information -- so prices cannot be fully revealing. Equilibrium prices are therefore only
\emph{partially} revealing: a strictly positive, unpredictable component must persist to compensate information
acquisition.
\end{proposition}

Unpredictability is thus not a defect of markets but a \emph{necessary feature} of an informational equilibrium: it
is the rent that pays for price discovery. The rational-expectations tradition \citep{lucas1978} formalises the same
endogeneity -- expectations are model-consistent inputs to the equilibrium -- and \citet{soros2013} names the
two-way feedback between beliefs and prices ``reflexivity.'' The reflexivity described here is closely related to the
evolutionary dynamics of Lo's Adaptive Markets Hypothesis \citep{lo2004}, where strategies compete, reproduce, and
die, generating a time-varying, context-dependent level of predictability. \citet{black1986} draws the operational
corollary: markets need noise traders to be liquid, and that noise is exactly what keeps the residual component
unforecastable. The disanalogy with physics could not be sharper: here the act of predicting alters what is
predicted.

% ======================================================================
\section{Microstructure: the causal chain beneath a tick}
\label{sec:microstructure}

The strongest antidote to the phrase ``markets are random'' is market microstructure. At the level where prices
actually move, a tick is not a draw from an urn. It is the result of a market order consuming displayed liquidity, a
limit order improving the quote, a dealer revising inventory risk, a queue being cancelled, or an informed trader
masking demand inside noise. Stochastic models are useful because the modeller cannot observe all of these causes,
not because the causes are absent.

The canonical example is the Kyle model \citep{kyle1985}. A risky asset has terminal value $v$, an informed trader
observes $v$, noise traders submit order $u$, and the market maker observes only total order flow
\begin{equation}\label{eq:kyle-y}
y=x+u,
\end{equation}
where $x$ is the informed order. Competitive pricing gives
\begin{equation}\label{eq:kyle-price}
p(y)=\E[v\mid y].
\end{equation}
The price is therefore efficient relative to the market maker's filtration generated by $y$, but not relative to the
informed trader's larger filtration generated by $v$. The same price change can be unpredictable for the market
maker, predictable for the informed trader, and still causal all the way down. The Glosten--Milgrom adverse-selection
model makes the same point in a bid--ask setting: spreads compensate liquidity suppliers for trading against agents
with superior information \citep{glostenmilgrom1985}.

Microstructure also explains why high-frequency predictability does not automatically become a free lunch. Let
$q_t$ denote signed order flow and write a stylised impact decomposition
\begin{equation}\label{eq:impact}
\Delta p_{t+h}=G_h(q_t,\mathcal O_t)+\eta_{t+h},
\end{equation}
where $\mathcal O_t$ is the order-book state, $G_h$ is transient or permanent impact, and $\eta$ collects news and
unobserved order-flow shocks. Order flow is empirically persistent, but price changes can remain close to martingale
because market makers adjust quotes, liquidity providers shade depth, and the cost of crossing the spread consumes
the predictable component \citep{hasbrouck2007}. This is the microstructural version of
\eqref{eq:netalphazero}: raw predictability is converted into spreads, impact, and adverse-selection premia.

\paragraph{Impact as intervention.}
Equation~\eqref{eq:impact} is not only predictive; it is causal. Submitting an order is an intervention on the book.
If a forecast recommends position $q$, the realised return is no longer drawn from the passive law
$\mathcal L(\Delta p\mid\mathcal O_t)$ but from the interventional law
\begin{equation}\label{eq:micro-do}
\mathcal L(\Delta p\mid\mathcal O_t,\operatorname{do}(q)).
\end{equation}
This is why high-frequency alpha is fragile: the same action that monetises the signal changes the queue, widens the
spread faced by the trader, reveals information, and invites adverse selection. Market impact is the microstructure
form of reflexivity.

The stylised facts of returns are therefore not paradoxical \citep{cont2001}. Returns can have weak linear
autocorrelation while absolute and squared returns display strong dependence; order flow can be predictable while
mid-price changes are close to conditionally mean-zero; liquidity can be forecastable while the profit from using
that forecast is competed into the bid--ask spread. The market is not a generator of metaphysical randomness. It is
a mechanism that hides causes behind aggregation and then prices access to those causes.

\begin{remark}[Why microstructure belongs in the argument]
Without microstructure, the essay risks sounding metaphysical: prices are said to be causal, but the causal chain is
not shown. Microstructure supplies the missing bridge. It turns the abstract sequence
``information $\to$ order $\to$ execution $\to$ price'' into models of quotes, queues, inventory, adverse selection,
and impact.
\end{remark}

% ======================================================================
\section{Knightian uncertainty: beyond probability}
\label{sec:knight}

The third and subtlest reason ``randomness'' misleads is that classical randomness is a draw from a \emph{fixed,
known} law -- an urn whose composition one cannot see but which does not change. Markets violate the premise. The
data-generating process drifts, regimes switch, and the future contains states never yet sampled. This is the
distinction of \citet{knight1921} between \emph{risk} (a known distribution) and \emph{uncertainty} (an unknown or
unstable one), and of \citet{keynes1936} on the irreducibly non-probabilistic character of long-horizon
expectation: ``we simply do not know.''

\begin{definition}[Ambiguity]\label{def:ambig}
A decision maker faces \emph{ambiguity} (Knightian uncertainty) when beliefs are represented not by a single prior
$\Prob$ but by a set of priors $\mathcal{P}$. Under the maxmin criterion \citep{gilboaschmeidler1989}, an action $a$
is evaluated by
\begin{equation}\label{eq:maxmin}
a \;\longmapsto\; \min_{\Prob\in\mathcal{P}} \E^{\Prob}\!\left[u(a)\right],
\end{equation}
the worst case over the entertained models.
\end{definition}

\paragraph{A concrete regime-switching illustration.}
Consider a simple two-regime setting: a ``normal'' regime $\mathcal{N}$ with return distribution $\mathcal{P}_{\mathcal{N}}$ and a ``crisis'' regime $\mathcal{C}$ with $\mathcal{P}_{\mathcal{C}}$. 
The agent does not know which regime will hold over the next period; she only knows that the true $\Prob$ is either $\mathcal{P}_{\mathcal{N}}$ or $\mathcal{P}_{\mathcal{C}}$, but cannot assign a probability to either. 
Her set of priors is $\mathcal{P} = \{\mathcal{P}_{\mathcal{N}}, \mathcal{P}_{\mathcal{C}}\}$. 
Under the maxmin criterion \citep{gilboaschmeidler1989}, she evaluates a trading strategy with payoff $X$ by
\[
\min_{\Prob\in\{\mathcal{P}_{\mathcal{N}},\mathcal{P}_{\mathcal{C}}\}} \E^{\Prob}[X].
\]
If instead she has a reference model $\Prob_{\mathrm{ref}}$ (e.g., a long-run average) but fears misspecification, the Hansen--Sargent robust control approach uses a penalty:
\[
\min_{\Q}\left\{ \E^{\Q}[ -X ] + \theta D(\Q\|\Prob_{\mathrm{ref}}) \right\},
\]
where $\theta$ is a confidence parameter. 
The set of plausible models becomes a relative-entropy ball $\{\Q: D(\Q\|\Prob_{\mathrm{ref}}) \le \eta\}$, which shrinks to the singleton $\Prob_{\mathrm{ref}}$ as $\theta\to\infty$. 
This bridges the Knightian set $\mathcal{P}$ and the entropy penalties of \S\ref{sec:relent}.

\paragraph{Dynamic ambiguity.}
In multiperiod markets the set of priors must also be dynamically coherent. If the modeller can choose one prior
today and an unrelated prior tomorrow, maxmin preferences become time inconsistent. Dynamic multiple-prior models
therefore impose stability under conditioning and pasting -- often called rectangularity -- so that the ambiguity set
is compatible with the filtration \citep{chenepstein2002}. For finance this matters because stress testing is not a
list of static scenarios. It is a family of conditional transition laws. A credible ambiguity set must say not only
which crises are possible, but how beliefs update as prices, liquidity, and balance sheets evolve.

The empirical reality of ambiguity is the Ellsberg paradox \citep{ellsberg1961}: agents strictly prefer bets on
known odds to bets on unknown odds, a preference no single prior can rationalise. Its pricing consequence is
structural. When the relevant model is not a point but a set $\mathcal{P}$, the single-prior machinery of
Theorem~\ref{thm:ftap} -- which fixes $\Prob$ up to equivalence -- gives way to a \emph{set} of martingale measures
$\mathcal{Q}$, and assets are priced by an interval rather than a point:
\begin{equation}\label{eq:bounds}
\underline{\pi}(X)=\inf_{\Q\in\mathcal{Q}}\E^{\Q}[X]
\;\le\;
\overline{\pi}(X)=\sup_{\Q\in\mathcal{Q}}\E^{\Q}[X],
\end{equation}
the sub- and super-replication bounds; robust-control formulations \citep{hansensargent2008} reach the same
conclusion from a max-min over models. This is the tier Corollary~\ref{cor:trichotomy}(iii) flagged: where there is
no fixed $\Prob$, there is no probability to be unpredictable \emph{within}. The uncertainty is not aleatory but
deep, and it is the part of market behaviour -- regime breaks, crises, structural change -- that resists not only
prediction but even probabilistic description. Keynes and Knight, not Kolmogorov. As shown in Section~\ref{sec:relent},
such Knightian ambiguity can be quantified by a set of priors expressed as a relative-entropy ball around a
reference model, making the distinction between risk (small ball) and deep uncertainty (large ball) a matter of
degree.

% ======================================================================
\section{Entropy: the measure of unpredictability}
\label{sec:entropy}

We have spoken of unpredictability qualitatively -- the martingale property \eqref{eq:mart}, the trichotomy of
Corollary~\ref{cor:trichotomy}. Information theory makes it a \emph{number}: Shannon's axioms single out, up to a
choice of base, a unique measure of the uncertainty carried by a distribution \citep{shannon1948}.

See \citet{shannon1948} and \citet{coverthomas2006}.
\begin{definition}[Entropy, conditional entropy, mutual information]
For a discrete random variable $X$ with law $(p_i)$, the \emph{entropy} is $H(X)=-\sum_i p_i\log p_i$; for a density
$f$, the \emph{differential entropy} is $h(X)=-\int f\log f$. Given a sub-$\sigma$-algebra $\F$ (an information
set), the \emph{conditional entropy} $H(X\mid\F)$ is the residual uncertainty in $X$ once $\F$ is known, and the
\emph{mutual information}
\begin{equation}\label{eq:mutualinfo}
I(X;\F)\;=\;H(X)-H(X\mid\F)\;\ge\;0
\end{equation}
measures how much $\F$ reduces it, with $I(X;\F)=0$ if and only if $X$ is independent of $\F$.
\end{definition}

Two facts fix the vocabulary. \emph{Conditioning cannot increase entropy:} $H(X\mid\F)\le H(X)$ -- information
weakly reduces unpredictability, the entropic shadow of the filtration $\Filt$ of \S\ref{sec:setup}. And among all
distributions of a given variance $\sigma^2$, the Gaussian uniquely \emph{maximises} differential entropy, at
$h=\tfrac12\log(2\pi e\sigma^2)$: the bell curve the central limit theorem makes universal is also the
least-structured law consistent with a second moment -- maximal unpredictability at fixed scale.

\paragraph{Forecastability is conditional information.}
This yields a precise information-theoretic statement, but it must be stated carefully. Let $r_k=\Delta X_k$ be the
next return and $\F_{k-1}$ the present information. The total informational content of the past for the next return is
$I(r_k;\F_{k-1})$. Mean forecastability is only one functional of that conditional law: is
$\E[r_k\mid\F_{k-1}]$ non-trivial? Scale forecastability is another: is $\Var(r_k\mid\F_{k-1})$ non-trivial? Higher
conditional cumulants add still more channels. The no-arbitrage martingale property \eqref{eq:mart} closes the
arbitrageable mean channel under the pricing measure, but it does not force $I(r_k;\F_{k-1})$ to be zero.

\begin{proposition}[The entropic trichotomy]\label{prop:entropic}
Market efficiency neutralises the \emph{predictable mean component} of risk-adjusted returns; it does not eliminate
all mutual information between future and past. The innovation $\Delta M_k$ has conditional mean zero yet may have a
forecastable conditional variance, jump intensity, liquidity state, tail thickness, or correlation regime. Its
conditional law is therefore structured even when its conditional mean is not.
\end{proposition}

\begin{proof}
The martingale property \eqref{eq:mart} sets $\E^{\Q}[r_k\mid\F_{k-1}]$ equal to the risk-free drift, so the
mean-channel mutual information vanishes. It places no constraint on the conditional variance, which may be
$\F_{k-1}$-measurable and non-constant (as in \eqref{eq:mds-vol}); the same holds for higher conditional cumulants.
Hence $I(r_k;\F_{k-1})$ can be strictly positive through moments above the first, which is the claim.
\end{proof}

\paragraph{The entropy rate.}
For a stationary return process the irreducible unpredictability per period is the \emph{entropy rate}
\begin{equation}\label{eq:entropyrate}
h\;=\;\lim_{n\to\infty}\frac1n\,H(r_1,\dots,r_n)\;=\;\lim_{n\to\infty} H\!\left(r_n\mid r_1,\dots,r_{n-1}\right),
\end{equation}
and the Shannon--McMillan--Breiman theorem \citep{coverthomas2006} gives it operational meaning: for an ergodic
process $-\tfrac1n\log p(r_1,\dots,r_n)\to h$ almost surely, so $h$ is the asymptotic rate at which the probability
of the \emph{realised} path decays -- the typical ``surprise per period.'' An i.i.d.\ fair signal attains the
maximal rate; predictable structure lowers it.

\paragraph{One entropy across the hierarchy.}
The same quantity unifies the kinds of unpredictability catalogued in the companion essay. \emph{Dynamical:} the
Kolmogorov--Sinai entropy $h_{\mathrm{KS}}(T)$ measures the unpredictability of a deterministic map; for chaotic
systems (positive Lyapunov exponents) it is positive, and Pesin's formula \citep{pesin1977} equates it to the sum
of those exponents, $h_{\mathrm{KS}}=\int\!\sum_{\lambda_i>0}\lambda_i\dd\mu$ -- so the doubling map's exponent
$\log 2$ \emph{is} its entropy rate. \emph{Algorithmic:} Brudno's theorem \citep{brudno1983} shows that for an
ergodic system the Kolmogorov-complexity rate of almost every orbit equals $h_{\mathrm{KS}}$, so the statistical
(Shannon), algorithmic (Kolmogorov), and dynamical (KS) measures of unpredictability coincide. \emph{Thermodynamic
and quantum:} the epistemic ignorance of a many-body system is its Gibbs entropy $-k_B\sum_i p_i\log p_i$, and the
ontic indeterminism of a measurement on a mixed state is the von Neumann entropy $-\Tr(\rho\log\rho)$. Entropy is the
common currency of unpredictability -- classical and quantum, deterministic and stochastic.

% ======================================================================
\section{Relative entropy: pricing the \texorpdfstring{$\Prob$--$\Q$}{P--Q} wedge}
\label{sec:relent}

If entropy measures unpredictability, \emph{relative} entropy measures the distance between probability
descriptions, and it is the exact instrument for the $\Prob$--$\Q$ wedge of \S\ref{sec:pq}.

\begin{definition}[Relative entropy; \citealp{kullbackleibler1951,coverthomas2006}]
For $\Q\ll\Prob$ the \emph{relative entropy} (Kullback--Leibler divergence) is
\begin{equation}\label{eq:reldef}
D(\Q\,\|\,\Prob)\;=\;\E^{\Q}\!\left[\log\frac{\dd\Q}{\dd\Prob}\right]\;\ge\;0,
\end{equation}
with equality if and only if $\Q=\Prob$: the informational cost of using $\Q$ in place of the true $\Prob$.
\end{definition}

The qualitative ``wedge $=$ price of risk'' of \S\ref{sec:pq} now becomes an identity.

\begin{proposition}[The entropy of the wedge]\label{prop:wedgeentropy}
In the lognormal market \eqref{eq:gbmP}--\eqref{eq:gbmQ}, the risk-neutral measure $\Q$ satisfies
\begin{equation}\label{eq:wedgeentropy}
D(\Q\,\|\,\Prob)\;=\;\tfrac12\,\theta^2\,T\;=\;\tfrac12\left(\frac{\mu-r}{\sigma}\right)^{2}T,
\end{equation}
one half the squared Sharpe ratio times the horizon.
\end{proposition}

\begin{proof}
From \eqref{eq:girsanov}, $\log\frac{\dd\Q}{\dd\Prob}=-\theta W_T^{\Prob}-\tfrac12\theta^2T$. Under $\Q$ the process
$W_t^{\Q}=W_t^{\Prob}+\theta t$ is a Brownian motion, so $W_T^{\Prob}=W_T^{\Q}-\theta T$ and
$\E^{\Q}[W_T^{\Prob}]=-\theta T$. Therefore
$D(\Q\|\Prob)=\E^{\Q}\!\left[-\theta W_T^{\Prob}-\tfrac12\theta^2T\right]=\theta^2T-\tfrac12\theta^2T=\tfrac12\theta^2T$.
\end{proof}

The informational distance of the pricing measure from reality is thus governed by the market price of risk: the
larger the premium an asset commands, the further $\Q$ must sit from $\Prob$, and the more structure pricing must
strip away. The identity also explains a notorious difficulty. By Pinsker's inequality,
$\norm{\Q-\Prob}_{\mathrm{TV}}\le\sqrt{\tfrac12 D(\Q\|\Prob)}=\tfrac12\,\abs{\theta}\sqrt{T}$, so $\Prob$ and $\Q$ --
and hence the drift that distinguishes them -- become statistically distinguishable only over horizons of order
$\theta^{-2}$. With Sharpe ratios well below one, that is decades of data: the equity premium is far harder to
estimate than volatility, exactly the point of \citet{merton1980}. The wedge is real but informationally faint.

\paragraph{A principled choice of pricing measure.}
In an incomplete market Theorem~\ref{thm:ftap} leaves a whole convex set $\mathcal M$ of equivalent martingale
measures, and \S\ref{sec:pq} noted that practitioners select among them ``by fiat.'' Relative entropy supplies a
canonical selection: the \emph{minimal entropy martingale measure} \citep{frittelli2000}
\begin{equation}\label{eq:memm}
\Q^{*}\;=\;\arg\min_{\Q\in\mathcal M}\;D(\Q\,\|\,\Prob),
\end{equation}
the martingale measure that distorts reality least, dual to maximising expected exponential utility. Of all the
instrumentally admissible probabilities, entropy selects the one of minimal informational violence.

\paragraph{Risk and ambiguity as entropy penalties.}
The same divergence quantifies attitudes to risk and to model uncertainty, closing the loop with \S\ref{sec:knight}.
The \emph{entropic risk measure} admits the dual representation \citep{follmerschied2016}
\begin{equation}\label{eq:entropicrisk}
\rho_\gamma(X)\;=\;\frac1\gamma\log\E^{\Prob}\!\left[e^{-\gamma X}\right]
\;=\;\sup_{\Q\ll\Prob}\Bigl\{\E^{\Q}[-X]-\tfrac1\gamma D(\Q\,\|\,\Prob)\Bigr\},
\end{equation}
a worst-case expected loss in which each alternative model $\Q$ is penalised by its relative entropy to the reference
$\Prob$, the risk aversion $\gamma$ trading pessimism against plausibility. The robust-control / multiplier
preferences of \citet{hansensargent2008} use the identical penalty,
$\min_{\Q}\{\E^{\Q}[\text{loss}]+\vartheta\,D(\Q\|\Prob)\}$, so the Knightian set of priors $\mathcal P$ of
\S\ref{sec:knight} acquires a quantitative size -- a relative-entropy ball around the reference model. And when only
partial information is available -- a few moments, or a strip of option prices -- the least-prejudiced model is the
one of \emph{maximum entropy} (equivalently, minimum relative entropy to a prior) consistent with the constraints
\citep{jaynes1957}: the principled response to the high-dimensional ignorance of \S\ref{sec:setup}, which in
derivatives recovers the risk-neutral density from quoted prices \citep{buchenkelly1996,stutzer1996}.

% ======================================================================

\section{An empirical protocol: how the thesis can be falsified}
\label{sec:empirical-protocol}

A theory of market hardness should not end as philosophy. It should imply tests. The relevant empirical question is
not whether one can find an in-sample pattern -- one always can in a high-dimensional search space -- but whether a
signal survives the sequence
\begin{equation}\label{eq:empirical-pipeline}
\begin{aligned}
\text{prediction}
&\longrightarrow \text{risk adjustment}
\longrightarrow \text{costs and impact}\\
&\longrightarrow \text{capacity}
\longrightarrow \text{out-of-sample survival}
\longrightarrow \text{live reflexive decay}.
\end{aligned}
\end{equation}
The essay's claim would be weakened by a public, scalable, persistent, riskless, net-of-cost alpha. It is not
weakened by a small premium, a volatility forecast, a private-information rent, a capacity-limited effect, or a
pattern that disappears after costs and deployment.

For a return forecast $\widehat m_t$ and realised excess return $r_{t+1}$, a minimal protocol separates three
objects:
\begin{align}
\text{raw predictive content:}\qquad
& R^2_{\mathrm{OS}}
=1-\frac{\sum_t (r_{t+1}-\widehat m_t)^2}{\sum_t(r_{t+1}-\bar r_{t}^{\mathrm{hist}})^2},
\label{eq:ros}\\
\text{economic value:}\qquad
& \Delta U
=\frac{1}{T}\sum_t u\!\bigl(W_{t+1}(\widehat m_t)\bigr)
-\frac{1}{T}\sum_t u\!\bigl(W_{t+1}(\text{benchmark})\bigr),
\label{eq:utilitygain}\\
\text{net tradability:}\qquad
& \widehat{\alpha}^{\mathrm{net}}
=\widehat{\alpha}^{\mathrm{gross}}-\text{spread}-\text{fees}-\text{borrow}-\text{impact}-\text{financing}.
\label{eq:nettradability}
\end{align}
A signal with positive $R^2_{\mathrm{OS}}$ but negative \eqref{eq:nettradability} confirms, rather than refutes, the
thesis: prediction exists, but exploitability is competed into costs, risk, impact, or capacity.

The protocol must also correct for model search. In modern finance, a researcher rarely tests one strategy; she tests
hundreds of transformations, horizons, assets, filters, labels, and hyperparameters. The relevant null is therefore
not the performance of one selected backtest, but the performance of the best selected backtest after data snooping.
White's Reality Check and Hansen's superior predictive ability test are classical corrections for this problem
\citep{white2000,hansen2005}; more recent financial machine-learning practice adds purged cross-validation,
embargoes, walk-forward evaluation, and deflated performance statistics to prevent leakage and selection bias
\citep{lopezdeprado2018}.

The falsification standard is deliberately severe. A pattern counts as a challenge only if it is public,
implementable, risk-adjusted, statistically robust after multiple-testing correction, capacity-aware, and persistent
after the capital that follows it changes the market. To make this concrete: a claim that the thesis is false would
require a strategy that, out-of-sample, achieves a large positive net Sharpe ratio after all costs, impact, funding,
and borrow constraints, across genuinely independent evaluation periods, with significance assessed after correction
for the number of strategies tried.

The sample-size arithmetic explains why the bar is high. Under the optimistic iid approximation, testing a true net
annual Sharpe $\SR$ against zero at two-sided level $\alpha$ with power $1-\beta$ requires roughly
\begin{equation}\label{eq:power-sharpe}
Y\gtrsim
\left(\frac{z_{1-\alpha/2}+z_{1-\beta}}{\SR}\right)^2
\end{equation}
independent years. For $\SR=0.8$, $\alpha=10^{-6}$, and $80\%$ power, this gives about $51$ years, not a dozen. A
12-year requirement corresponds roughly to the much weaker $5\%$ test with the same Sharpe. Serial dependence,
non-normal tails, overlapping portfolios, hyperparameter search, and live capital feedback all make the true burden
larger; deflated-Sharpe and multiple-testing corrections target precisely these distortions \citep{baileylopezdeprado2014,harveyliuzhu2016}. This high bar is exactly the bar implied by a reflexive market: anything weaker is not evidence of ontic
randomness or its absence; it is only evidence about a particular information set and trading technology.

% ======================================================================

\section{Synthesis and consistency with the empirical record}
\label{sec:synthesis}

\paragraph{A taxonomy of market hardness.}
The sources, the filtration/cost qualifications, and the return decomposition combine into the following picture,
which places markets entirely in the \emph{epistemic} column of the companion essay while adding mechanisms absent
from passive physical systems.

\begin{center}
\renewcommand{\arraystretch}{1.35}
\begin{tabular}{@{}p{0.215\textwidth}p{0.32\textwidth}p{0.385\textwidth}@{}}
\toprule
\textbf{Source} & \textbf{Nature} & \textbf{Formal locus}\\
\midrule
High-dimensional partial observation &
Classical ignorance of many coupled, unobserved variables; shared with statistical physics &
Probabilistic modelling of an aggregate (\S\ref{sec:setup}); the modeller's $\Prob$.\\
Finite-sample learnability &
A signal can exist but be too faint, unstable, crowded, or costly to estimate and deploy &
Drift horizon \eqref{eq:drift-horizon}, model-selection bias \eqref{eq:winner-bias}, capacity \eqref{eq:capacity-objective}.\\
Reflexivity &
Forecast feeds back into the forecasted system; \emph{market-specific} &
Self-defeating predictability and the equilibrium martingale (Prop.~\ref{prop:reflexive}--\ref{prop:gs};
Eq.~\eqref{eq:mart}).\\
Knightian uncertainty &
No fixed law; non-stationarity, regime change; \emph{beyond probability} &
Set of priors $\mathcal{P}$, maxmin \eqref{eq:maxmin}, price bounds \eqref{eq:bounds}.\\
\midrule
\multicolumn{3}{@{}l}{\emph{Return components} (Corollary~\ref{cor:trichotomy}):}\\
Martingale innovation & Arbitrageable; unpredictable \emph{by construction} & $\Delta M$, $\E[\Delta M\mid\F_{k-1}]=0$.\\
Risk premium & Predictable but not a free lunch & $\Delta A$, reabsorbed under $\Q$.\\
Regime/tail & Not probabilizable & Outside any fixed $\Prob$.\\
\bottomrule
\end{tabular}
\end{center}

\paragraph{The entropy ledger.}
Collected as quantities, the thesis reads as a single ledger of entropies and information constraints.

\begin{center}
\renewcommand{\arraystretch}{1.3}
\begin{tabular}{@{}p{0.45\textwidth}p{0.46\textwidth}@{}}
\toprule
\textbf{Phenomenon} & \textbf{Information quantity or constraint}\\
\midrule
Irreducible unpredictability per period & entropy rate $h$ \eqref{eq:entropyrate}\\
Forecastability of returns & mutual information $I(r_k;\F_{k-1})$ \eqref{eq:mutualinfo}\\
Efficiency (arbitrageable channel) & predictable mean component neutralised (Prop.~\ref{prop:entropic})\\
Learnability of alpha & sample-size and model-selection constraints \eqref{eq:drift-horizon}--\eqref{eq:winner-bias}\\
Volatility clustering & second-moment information -- forecastable, not free\\
The $\Prob$--$\Q$ wedge / price of risk & relative entropy $D(\Q\|\Prob)=\tfrac12\theta^2T$ (Prop.~\ref{prop:wedgeentropy})\\
Choice of pricing measure & minimal entropy martingale measure \eqref{eq:memm}\\
Risk aversion / Knightian ambiguity & entropy-penalised worst case \eqref{eq:entropicrisk}\\
Modelling under ignorance & maximum entropy / minimum relative entropy\\
\bottomrule
\end{tabular}
\end{center}

\begin{table}[h]
\centering
\caption{Illustrative one-year calibration. The numbers are deliberately rough: equity premium $\mu-r=5\%$ and volatility $\sigma=15\%$.}
\label{tab:calibration}
\begin{tabular}{@{}ll@{}}
\toprule
Quantity & Value \\
\midrule
Market price of risk & $\theta=(\mu-r)/\sigma\approx 0.33$ \\
Relative entropy wedge & $D(\Q\|\Prob)=\frac12\theta^2T\approx 0.055$ nats \\
Same wedge in bits & $0.055/\log 2\approx 0.079$ bits \\
Pinsker upper bound & $\|\Q-\Prob\|_{\mathrm{TV}}\le\frac12|\theta|\sqrt{T}\approx 0.165$ \\
\bottomrule
\end{tabular}
\end{table}

\paragraph{Consistency with the empirical record.}
We do not re-estimate these facts here; the point is that the decomposition is consistent with well-established
empirical findings. On the \emph{unpredictable} side, returns
display weak or economically negligible serial correlation at horizons relevant to trading, the empirical core of the
efficient-market hypothesis \citep{fama1970}, with the random-walk idea tracing to \citet{bachelier1900}. On the
\emph{predictable but risky} side, conditional variance is strongly forecastable -- volatility clustering, captured
by the ARCH/GARCH family \citep{engle1982,bollerslev1986}; average excess returns carry cross-sectional premia
associated with market, size, and value exposures \citep{sharpe1964,famafrench1993}; and valuation ratios forecast
long-horizon returns \citep{campbellshiller1988}. Each is a forecastable $\Delta A$ that compensates risk, not a
violation of the martingale \eqref{eq:mart} under $\Q$. On the \emph{structured but non-Gaussian} side, prices exhibit
heavy tails incompatible with the normal law \citep{mandelbrot1963}, excess volatility relative to simple dividend
models \citep{shiller1981}, and liquidity that depends on the presence of noise traders \citep{black1986}. None of
this is the fingerprint of ontic chance; all of it is the fingerprint of structured ignorance, endogenous feedback,
limited learnability, and a shifting law.

% ======================================================================

\section{Conclusion}
\label{sec:conclusion}

The thesis, made precise, is as follows. Financial markets are \emph{epistemically hard}, not \emph{ontically
random}. The Born-rule notion of chance has no direct foothold in a price formed by causal economic interaction.
Probability enters finance for a different reason: it is the calculus of incomplete information, strategic
interaction, risk compensation, finite-sample learning, and model uncertainty.

The first formal lesson is that probability in finance is \emph{instrumental}. The field's central computation runs
under a risk-neutral measure $\Q\ne\Prob$ that is selected because it prices assets conveniently, not because anyone
believes it is the real-world data-generating law. In the lognormal benchmark the informational size of the wedge is
not a metaphor but a number,
\begin{equation*}
D(\Q\|\Prob)=\frac12\theta^2T,
\end{equation*}
one half the squared Sharpe ratio times the horizon. The price of risk is literally an informational distortion.

The second lesson is that no-arbitrage, efficiency, predictability, learnability, and exploitability must not be
collapsed. No-arbitrage gives a martingale representation under an equivalent pricing measure. Efficiency is an
information-set-relative empirical claim. Predictability is a property of a conditional law. Learnability is a
finite-sample and computational property of an estimator. Exploitability is a net-of-cost implementation claim.
Predictable premia, volatility clustering, microstructure order-flow effects, private information, and long-horizon
valuation signals can all coexist with a market that contains no scalable free lunch. The right distinction is not
``predictable versus random'' but \emph{predictable and compensated}, \emph{predictable but unlearnable},
\emph{predictable and inaccessible}, \emph{predictable but self-defeating}, or \emph{unforecastable given the
available filtration}.

The third lesson is causal. A market is not a passive dynamical system observed from outside. It is a strategic
system whose participants act on forecasts, and those actions change prices, liquidity, spreads, and future
forecasts. Microstructure makes this concrete: each tick arises from orders, queues, inventory, adverse selection,
and impact. Reflexivity makes it general: a public edge erodes because it is traded. Knightian uncertainty makes it
harder still: sometimes the relevant law itself shifts, so the difficulty is not estimating a fixed distribution but
surviving when no single prior is adequate.

Entropy supplies the common measurement language. Conditional entropy is residual unpredictability; mutual
information is forecastability; relative entropy is the distance between pricing and physical measures; entropy
penalties quantify robust control and ambiguity; maximum entropy models are the disciplined response to partial
constraints. The entropic view also clarifies why markets look deceptively random. The first-moment channel is small
and competed away, the volatility channel remains informative, the risk-premium channel is faint, and the tail/regime
channel is structurally unstable.

The practical moral for quantitative finance is therefore not that markets are random, nor that they are perfectly
predictable. The discipline is the \emph{management of limited predictability}: building richer filtrations, testing
signals out of sample, paying for information only when it survives costs and impact, distinguishing predictive
correlation from intervention-stable causality, respecting finite-sample learnability, and hedging against model
uncertainty when the law itself may change. Markets do not play dice. They keep secrets, react to being read, and
sometimes change the rules. That is not randomness; it is a harder and more interesting problem.

% ======================================================================
\bibliographystyle{plainnat}
\bibliography{finance_reviewed_extended_v3_refs}

\end{document}